\documentclass{article}
\usepackage[english]{babel}
\usepackage{amsmath}
\usepackage{amssymb}
\usepackage{amsthm}
\usepackage{mathrsfs}
\usepackage{hyperref}
\usepackage{tikz}
\usetikzlibrary{arrows,automata,positioning, arrows.meta, calc}
\usepackage{delarray}
\usepackage{enumerate}
\usepackage{tabularx}
\usepackage{algorithm}
\usepackage[noend]{algpseudocode}
\usepackage{environ}
\usepackage{authblk}
\usepackage{stmaryrd}

\tikzset{configuration/.style = {state, rectangle, minimum height=0.5cm},
}

\let\emptyset\varnothing

\newcommand{\B}{\mathbb{B}}
\newcommand{\N}{\mathbb{N}}

\newcommand{\dom}{\mathrm{dom}}

\newcommand{\NP}{$\text{NP}$}
\newcommand{\coNP}{$\text{coNP}$}
\newcommand{\NPcoNP}{$\text{NP}^\text{coNP}$}
\newcommand{\FNP}{$\text{FNP}$}
\newcommand{\FcoNP}{$\text{FcoNP}$}
\newcommand{\FNPcoNP}{$\text{FNP}^\text{coNP}$}

\makeatletter
\newcommand{\problemtitle}[1]{\gdef\@problemtitle{\hspace*{-3.5mm}\scalebox{.75}{$\blacktriangleright$} #1}}
\newcommand{\probleminput}[1]{\gdef\@probleminput{#1}}
\newcommand{\problemquestion}[1]{\gdef\@problemquestion{#1}}
\NewEnviron{dproblem}{
  \problemtitle{}\probleminput{}\problemquestion{}
  \BODY
  \par\addvspace{.5\baselineskip}
  \noindent
  \begin{tabularx}{\textwidth}{@{\hspace{\parindent}} l X c}
    \multicolumn{2}{@{\hspace{\parindent}}l}{\@problemtitle} \\
    \textbf{Input:} & \@probleminput \\
    \textbf{Question:} & \@problemquestion
  \end{tabularx}
  \par\addvspace{.5\baselineskip}
}

\makeatletter
\newcommand{\fproblemtitle}[1]{\gdef\@fproblemtitle{\hspace*{-3.5mm}\scalebox{.75}{$\blacktriangleright$} #1}}
\newcommand{\fprobleminput}[1]{\gdef\@fprobleminput{#1}}
\newcommand{\fproblemoutput}[1]{\gdef\@fproblemoutput{#1}}
\NewEnviron{fproblem}{
  \fproblemtitle{}\fprobleminput{}\fproblemoutput{}
  \BODY
  \par\addvspace{.5\baselineskip}
  \noindent
  \begin{tabularx}{\textwidth}{@{\hspace{\parindent}} l X c}
    \multicolumn{2}{@{\hspace{\parindent}}l}{\@fproblemtitle} \\
    \textbf{Input:} & \@fprobleminput \\
    \textbf{Output:} & \@fproblemoutput
  \end{tabularx}
  \par\addvspace{.5\baselineskip}
}

\makeatletter
\newcommand{\fpproblemtitle}[1]{\gdef\@fpproblemtitle{\hspace*{-3.5mm}\scalebox{.75}{$\blacktriangleright$} #1}}
\newcommand{\fpprobleminput}[1]{\gdef\@fpprobleminput{#1}}
\newcommand{\fpproblempromise}[1]{\gdef\@fpproblempromise{#1}}
\newcommand{\fpproblemoutput}[1]{\gdef\@fpproblemoutput{#1}}
\NewEnviron{fpproblem}{
  \fpproblemtitle{}\fpprobleminput{}\fpproblempromise{}\fpproblemoutput{}
  \BODY
  \par\addvspace{.5\baselineskip}
  \noindent
  \begin{tabularx}{\textwidth}{@{\hspace{\parindent}} l X c}
    \multicolumn{2}{@{\hspace{\parindent}}l}{\@fpproblemtitle} \\
    \textbf{Input:} & \@fpprobleminput \\
    \textbf{Promise:} & \@fpproblempromise \\
    \textbf{Output:} & \@fpproblemoutput 
  \end{tabularx}
  \par\addvspace{.5\baselineskip}
}

\newtheorem{example}{Example}
\newtheorem{theorem}{Theorem}
\newtheorem{claim}{Claim}

\begin{document}
\title{Optimising attractor computation in Boolean automata networks}
\author[1,2]{K{\'e}vin Perrot}
\author[1]{Pac{\^o}me Perrotin}
\author[1]{Sylvain Sen{\'e}}
\affil[1]{Aix Marseille Univ., Univ. de Toulon, CNRS, LIS, France.} 
\affil[2]{Univ. C{\^o}te d'Azur, CNRS, I3S, France.} 
\date{}

\maketitle

\begin{abstract}
This paper details a method for optimising the size of Boolean
automata networks in order to compute
their attractors under the parallel update schedule. This method relies
on the formalism of modules introduced recently that allows for (de)composing
such networks.
We discuss the practicality of this method by exploring examples. We also
propose results that nail the complexity of most parts of the process, while
the complexity of one part of the problem is left open.
\end{abstract}

\section{Introduction}

Boolean automata networks (BANs) are studied for their capacity to succintly
expose the complexity that comes with the composition of simple entities into
a network. They belong to a wide family of systems which include
cellular automata and neural networks, and can be described as cellular
automata with arbitrary functions and on arbitrary graph structures.

Understanding and predicting the dynamics of computing with BANs has been a
focus of the scientific community which studies them, in particular since
their applications include the modelling of gene regulatory
networks~\cite{J-Kauffman1969,J-Thomas1973,J-Mendoza1998,J-Davidich2008,J-Demongeot2010}.
In those applications, fixed points of a BAN are often viewed as
cellular types and limit cycles as biological
rhythms~\cite{J-Kauffman1969,J-Thomas1973}.
It follows that most biological studies
relying on BANs require the complete computation of their dynamics to propose
conclusions.
The complete computation of the dynamics of BANs is an exponentially costly
process. Indeed, for $n$ the size of a BAN, the size of its dynamics
is precisely $2^n$. The dynamics of a BAN is usually partitionned in two sorts
of configurations: the recurring ones that are parts of attractors and
either belong to a limit cycle or are fixed points; the others that evolve
towards these attractors and belong to their attraction basins.
The questions of characterising, computing or counting those attractors from a
simple description of the
network have been
explored~\cite{C-Orponen1989,J-Aracena2008,J-Goles2008,J-Demongeot2012,C-Noual2012,J-Aracena2017},
and have been shown to be difficult
problems~\cite{C-Orponen1989,C-Orponen1992,C-Bridoux2019,U-Bridoux2020,C-Nous2020}.

In this paper, we propose a new method for computing the
attractors of a BAN under the parallel update schedule. For any input network,
this method generates another network which is possibly smaller and which is
guaranteed to possess attractors isomorphic to those of the input network.
Computing
the dynamics of this smaller network therefore takes as much time as needed
to compute the dynamics of the input networks, divided by some power of two.

This method uses tools and results developed in previous works by the
authors~\cite{C-Perrot2018,C-Perrot2020}.
These works involve adding inputs to BANs, in a generalisation called modules
that proposes in some cases the study of the computationnal capabilities of the
network as the computation in terms of the inputs.
In particular, a result states that two networks that have equivalent such
computations share isomorphic attractors.

Section~\ref{s:Definitions} starts by exposing all the definitions
needed to read this paper.
Section~\ref{s:PartA} explores the question of obtaining an acyclic module
(AM) from
a BAN. Section~\ref{s:PartB} explains how to extract so called output functions
from a module. Section~\ref{s:PartC} details how to generate a minimal module
from a set of output functions. Finally Section~\ref{s:PartD} shows the final
step of the method, which implies constructing a BAN out of an AM
and computing its dynamics. Each section explores complexity results
of the different parts of the process, and details examples along the way.
An illustrative outline of the paper can be found in Figure~\ref{f:Plan}.

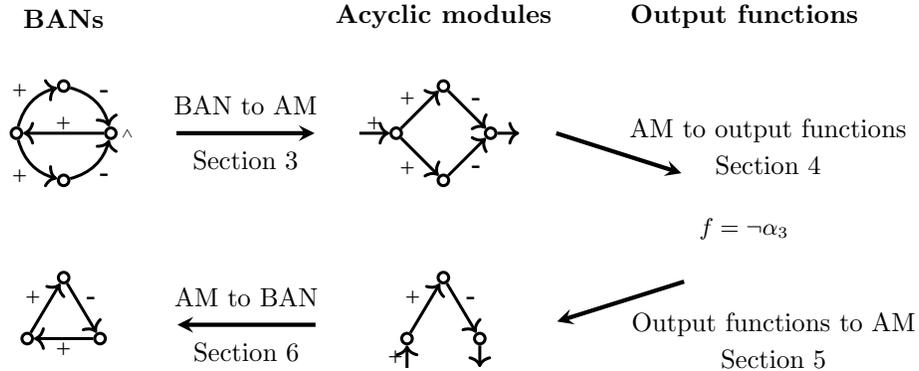
\begin{figure}[t!]
\label{f:Plan}
\begin{center}
\begin{tikzpicture} [->, line width=0.01mm,auto,node distance=0.4cm,line width=0.4mm]

        \node (CENTER) at (0, 0) {};


        \node (BANCENTER) [left =4cm of CENTER] {};

        \node (BAN1CENTER) [above =1cm of BANCENTER] {};

		\node[state,inner sep=1pt,minimum size=4pt] (BAN1A) [left =of BAN1CENTER] {};
		\node[state,inner sep=1pt,minimum size=4pt] (BAN1B) [above =of BAN1CENTER] {};
		\node[state,inner sep=1pt,minimum size=4pt] (BAN1C) [below =of BAN1CENTER] {};
		\node[state,inner sep=1pt,minimum size=4pt] (BAN1D) [right =of BAN1CENTER] {};

		\path (BAN1A) edge [bend left] node[above left=-3pt] {\tiny \textbf{+}} (BAN1B);
		\path (BAN1A) edge [bend right] node[below left=-3pt] {\tiny \textbf{+}} (BAN1C);
		\path (BAN1B) edge [bend left] node[above right=-3pt] {\small \textbf{-}} (BAN1D);
		\path (BAN1C) edge [bend right] node[below right=-3pt] {\small \textbf{-}} (BAN1D);
		\path (BAN1D) edge node[above=-2.5pt] {\tiny \textbf{+}} (BAN1A);

		\node [right of=BAN1D, xshift=-1.8mm] {\tiny \textbf{$\wedge$}};


        \node (BAN2CENTER) [below =1cm of BANCENTER] {};

		\node[state,inner sep=1pt,minimum size=4pt] (BAN2A)
          [below left =of BAN2CENTER, yshift =0.3cm] {};
		\node[state,inner sep=1pt,minimum size=4pt] (BAN2B) [above =of BAN2CENTER] {};
		\node[state,inner sep=1pt,minimum size=4pt] (BAN2C)
          [below right =of BAN2CENTER, yshift =0.3cm] {};

		\path (BAN2A) edge node[above left=-3pt] {\tiny \textbf{+}} (BAN2B);
		\path (BAN2B) edge node[above right=-3pt] {\small \textbf{-}} (BAN2C);
		\path (BAN2C) edge node[below=-2.5pt] {\tiny \textbf{+}} (BAN2A);

	\node [above =2.4cm of BANCENTER] {\textbf{BANs}};


        \node (MODCENTER) [right = 0.5cm of CENTER] {};

        \node (MOD1CENTER) [above =1cm of MODCENTER] {};

		\node[state,inner sep=1pt,minimum size=4pt] (MOD1A) [left =of MOD1CENTER] {};
		\node[state,inner sep=1pt,minimum size=4pt] (MOD1B) [above =of MOD1CENTER] {};
		\node[state,inner sep=1pt,minimum size=4pt] (MOD1C) [below =of MOD1CENTER] {};
		\node[state,inner sep=1pt,minimum size=4pt] (MOD1D) [right =of MOD1CENTER] {};

		\path (MOD1A) edge node[above left=-3pt] {\tiny \textbf{+}} (MOD1B);
		\path (MOD1A) edge node[below left=-3pt] {\tiny \textbf{+}} (MOD1C);
		\path (MOD1B) edge node[above right=-3pt] {\small \textbf{-}} (MOD1D);
		\path (MOD1C) edge node[below right=-3pt] {\small \textbf{-}} (MOD1D);
		\path node[left =of MOD1A] {} edge node[above=-2.5pt] {\tiny \textbf{+}} (MOD1A);
		\path [draw] (MOD1D) -- ++(0.38, 0);


        \node (MOD2CENTER) [below =1cm of MODCENTER] {};

		\node[state,inner sep=1pt,minimum size=4pt] (MOD2A)
          [below left =of MOD2CENTER, yshift =0.3cm] {};
		\node[state,inner sep=1pt,minimum size=4pt] (MOD2B) [above =of MOD2CENTER] {};
		\node[state,inner sep=1pt,minimum size=4pt] (MOD2C)
          [below right =of MOD2CENTER, yshift =0.3cm] {};

		\path (MOD2A) edge node[above left=-3pt] {\tiny \textbf{+}} (MOD2B);
		\path (MOD2B) edge node[above right=-3pt] {\small \textbf{-}} (MOD2C);
		\path node[below =of MOD2A, yshift=1mm] {} edge node[left=-2.5pt] {\tiny \textbf{+}} (MOD2A);
		\path [draw] (MOD2C) -- ++(0, -0.38);

	\node [above =2.4cm of MODCENTER] {\textbf{Acyclic modules}};


        \node (FUNCENTER) [right = 4.5cm of CENTER] {};

		\node (FUNC) at (FUNCENTER) {\small $f = \neg \alpha_3$};

	\node [above =2.4cm of FUNCENTER] {\textbf{Output functions}};


    \node (PARTATIP) [left =1.3cm of MOD1CENTER] {};
	\draw [>=stealth, ultra thick, draw]
		(BAN1CENTER) ++ (1.5cm, 0) --
        node (partAtag) [above=0.1 cm] {BAN to AM}
        (PARTATIP) {};
    \node [below=0.2 cm of partAtag] {Section~\ref{s:PartA}};

    \node (PARTBTIP) [above left =0.6cm of FUNCENTER] {};
	\draw [>=stealth, ultra thick, draw]
		(MOD1CENTER) ++ (1.5cm, 0) --
        node (partBtag) [above right=0 cm] {AM to output functions}
        (PARTBTIP) {};
    \node [below=-0.1 cm of partBtag] {Section~\ref{s:PartB}};

    \node (PARTCTIP) [right =1.1cm of MOD2CENTER] {};
	\draw [>=stealth, ultra thick, draw]
		(FUNCENTER) ++ (-0.8cm, -0.7cm) --
        node (partCtag) [below right=0 cm] {Output functions to AM}
        (PARTCTIP) {};
    \node [below=-0.1 cm of partCtag] {Section~\ref{s:PartC}};

    \node (PARTDTIP) [right =1.1cm of BAN2CENTER] {};
	\draw [>=stealth, ultra thick, draw]
		(MOD2CENTER) ++ (-1.7cm, 0) --
        node (partDtag) [above=0.1 cm] {AM to BAN}
        (PARTDTIP) {};
    \node [below=0.2 cm of partDtag] {Section~\ref{s:PartD}};
		
\end{tikzpicture}
\end{center}
\caption{
Illustration of the optimisation pipeline explored in this paper. Each arrow
corresponds to a part of the pipeline, and a section in this article.
}
\end{figure}


\section{Definitions} \label{s:Definitions}

\subsection{Boolean functions}

In this paper, we consider a Boolean function as any function
$f : \B^A \to \B$, for $A$ a finite set. An affectation $x$ of $f$ is a vector
in $\B^A$.
When considered as the input or output of a complexity problem, we encode
Boolean functions as Boolean circuits.
A \emph{Boolean circuit} of $f$ is an acyclic digraph in which nodes without
incoming edges are
labelled by an element in $A$, and every other node by a Boolean gate in
$\{\wedge, \vee, \neg\}$, with a special node marked as the output of the circuit.
The evaluation $f(x)$ is
computed by mapping $x$ to the input nodes of the circuit, and propagating
the evaluation along the circuit using the gates until the output node is
reached.


\subsection{Boolean automata networks and acyclic modules}

\subsubsection{Boolean automata networks}

BANs are composed of a set $S$ of automata.
Each automaton in $S$, or node, is at any time in a state in $\B$.
Gathering those isolated states into a vector of dimension $|S|$ provides us 
with a configuration of the network. 
More formally, a \emph{configuration} of $S$ over $\B$ is a vector in 
$\B^S$.
The state of every automaton is bound to evolve as a function of the
configuration of the entire network. Each node has a unique function, called
a local function, that is predefined and does not change over time. 
A \emph{local function} is thus a function $f$ defined as $f: \B^S \to 
\B$.
Formally, a BAN $F$ is a set that assigns a local
function $f_s$ over $S$ for every $s \in S$.

\begin{example}
\label{ex1-BAN-def}
Let $S_A = \{a, b, c, d\}$. Let $F_A$ be the BAN defined by
$f_a(x) = x_d$, $f_b(x) = f_c(x) = x_a$, and
$f_d(x) = \neg x_b \vee \neg x_c$.
The interaction digraph of this BAN is depicted in Figure~\ref{f:ex:BAN}
(left panel).
\end{example}

\begin{example}
\label{ex2-BAN-def}
Let $S_B = \{St, Sl, Sk, Pp, Ru, S9, C, C25, M, C^*\}$. Let $F_B$ be the BAN
defined by
$f_{St}(x) = \neg x_{St}$,
$f_{Sl}(x) = \neg x_{Sl} \vee x_{C^*}$,
$f_{Sk}(x) = x_{St} \vee \neg x_{Sk}$,
$f_{Pp}(x) = x_{Sl} \vee \neg x_{Pp}$,
$f_{Ru}(x) = f_{S9}(x) = \neg x_{Sk} \vee x_{Pp} \vee \neg x_C \vee \neg x_{C^*}$,
$f_C(x) = \neg x_{Ru} \vee \neg x_{S9} \vee \neg x_{Sl}$,
$f_{C25}(x) = \neg x_{Pp} \vee x_C$,
$f_{M}(x) = x_{Pp} \vee \neg x_C$,
and $f_{C^*}(x) = \neg x_{Ru} \vee \neg x_{S9} \vee x_{C25} \vee \neg x_M$.
The interaction digraph of this BAN is depicted in Figure~\ref{f:ex:BAN}
(right panel).
\end{example}

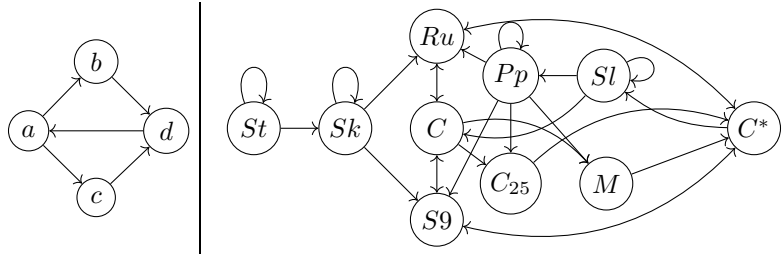
\begin{figure}[t!]
\begin{center}
\begin{minipage}{.2\textwidth}
\begin{tikzpicture} [->, node distance=0.5cm]

        \node (CENTER) at (0, 0) {};

        \node[state, inner sep=3pt,minimum size=0] (A) [left = of CENTER] {$a$};
        \node[state, inner sep=3pt,minimum size=0] (B) [above= of CENTER] {$b$};
        \node[state, inner sep=3pt,minimum size=0] (C) [below= of CENTER] {$c$};
        \node[state, inner sep=3pt,minimum size=0] (D) [right= of CENTER] {$d$};

        \path (A) edge (B) edge (C)
              (B) edge (D)
              (C) edge (D)
              (D) edge (A);
		
\end{tikzpicture}
\end{minipage}
\vline\quad
\begin{minipage}{.60\textwidth}
\begin{tikzpicture} [->, node distance=0.5cm]

        \node (CENTER) at (0, 0) {};

        \node[state, inner sep=2pt,minimum size=20pt] (C) [left = of CENTER] {$C$};
        \node[state, inner sep=2pt,minimum size=20pt] (Sk) [left= of C] {$Sk$};
        \node[state, inner sep=2pt,minimum size=20pt] (St) [left= of Sk] {$St$};
        \node[state, inner sep=2pt,minimum size=20pt] (Ru) [above= of C] {$Ru$};
        \node[state, inner sep=2pt,minimum size=20pt] (S9) [below= of C] {$S9$};
        \node[state, inner sep=2pt,minimum size=20pt] (Pp) [above=0.2cm of CENTER] {$Pp$};
        \node[state, inner sep=2pt,minimum size=20pt] (C25) [below=0.2cm of CENTER] {$C_{25}$};
        \node[state, inner sep=2pt,minimum size=20pt] (Sl) [right= of Pp] {$Sl$};
        \node[state, inner sep=2pt,minimum size=20pt] (M) [right= of C25] {$M$};
        \node[state, inner sep=2pt,minimum size=20pt] (C*) [right=3.5cm of C] {$C^*$};

        \path (St) edge (Sk)
              (Sk) edge (Ru) edge (S9)
              (Ru) edge[<->, bend left] (C*)
                   edge[<->] (C)
              (S9) edge[<->, bend right] (C*)
                   edge[<->] (C)
              (C)  edge (C25)
                   edge [bend left] (M)
              (Pp) edge (Ru) edge (S9) edge (C25) edge (M)
              (Sl) edge (Pp)
                   edge [bend left] (C)
              (C25) edge [bend left] (C*)
              (M) edge (C*)
              (C*) edge [bend left = 18] (Sl);
	     \draw (St) to [out=110, in=70, looseness=7] (St);
	     \draw (Sk) to [out=110, in=70, looseness=7] (Sk);
	     \draw (Pp) to [out=110, in=70, looseness=5] (Pp);
	     \draw (Sl) to [out=30, in=350, looseness=5] (Sl);
		
\end{tikzpicture}
\end{minipage}
\end{center}
\caption{
\label{f:ex:BAN}
On the left, the interaction digraph of $F_A$, as described in Example~\ref{ex1-BAN-def}.
On the right, the interaction digraph of $F_B$, as described in Example~\ref{ex2-BAN-def}.
}
\end{figure}

In the scope of this paper, BANs (and modules) are udpated according to
the parallel update schedule.
Formally, for $F$ a BAN and $x$ a configuration of $F$, the update of $x$ under
$F$ is denoted by configuration $F(x)$, and defined as for all $s$ in $S$,
$F(x)_s = f_s(x)$.

\begin{example}
\label{ex1-BAN-update}
Consider $F_A$ of Example~\ref{ex1-BAN-def}, and $x \in \B^{S_A}$ such that
$x = 1001$.
We observe that $F_A(x) = 1111$. Configurations $1000$ and $0111$ are recurring
and form a limit cycle of size $2$, 
as well as configurations $0000$, $0001$, $1001$, $1111$, $1110$ and $0110$
that form a limit cycle of size $6$.
\end{example}

\subsubsection{Interaction digraph}

BANs are usually represented by the influence that automata hold on each other.
As such the visual representation of a BAN is a digraph, called an
interaction digraph, whose nodes are the automata of the network,
and arcs are the influences that link the different automata. 
Formally, $s$ \emph{influences} $s'$ if and only if there exist two 
configurations $x,x'$ such that $f_{s'}(x) \neq f_{s'}(x')$
and for all $r$ in $S$, $r \neq s$ if and only if $x_r = x'_r$.

\subsubsection{Dynamics}

Finally, we define the \emph{dynamics} of a BAN $F$ as the digraph with $\B^S$ as its
set of vertices. There exists an edge from $x$ to $y$ if and only if $F(x) = y$.
Computing the dynamics of a BAN from the description of its local function
is an exponential process.
See~\cite{B-Robert1986} for a more throughout introduction to BANs and
related subjects.

\subsubsection{Modules}

Modules were first introduced in~\cite{C-Perrot2018}.
A module $M$ is a BAN with added inputs. It is defined on two sets: $S$
a set of automata, and $I$ a set of inputs, with $S \cap I = \emptyset$.
Similarly to standard BANs, we can define configurations as vectors in $\B^S$,
and we define input configurations as vectors in $\B^I$.
A local function of a module updates itself based on a configuration $x$
and an input configuration $i$, concatenated into one configuration.
Formally, a local function is defined from
$\B^{S \cup I}$ to $\B$. The module $M$ defines a local function for every node
$s$ in $S$.

\begin{example}
Let $M_e$ be the module defined on $S_e = \{p, q, r\}$ and $I = \{\alpha, \beta\}$,
such that $f_p(x) = x_\alpha$,
$f_q(x) = \neg x_p$,
and $f_r(x) = x_q \vee \neg x_\beta$.
\end{example}

We represent modules with an interaction digraph, in the same way as for BANs.
The interaction digraph of a module has added arrows that represent the
influence of the inputs over the nodes; for every node $s$ and every input
$\alpha$, the node $s$ of the interaction digraph has an ingoing arrow labelled
$\alpha$ if and only if $\alpha$ influences $s$, that is, there exists two 
input configurations $i,i'$ such that for all $\beta$ in $I$,
$\beta \neq \alpha$ if and only if
$i_\beta = {i'}_\beta$, and $x$ a configuration such that
$f_s(x\cdot i) \neq f_s(x\cdot i')$, where $\cdot$ denotes the concatenation
operator.

A module is \emph{acyclic}
if and only if its interaction digraph is cycle-free.

\subsubsection{Recursive wirings}

A recursive wiring over a module $M$ is defined by a partial function
$\omega : I \not\to S$. The result of such a wiring is denoted
$\circlearrowright_\omega M$, a module defined over sets $S$ and
$I \setminus \dom(\omega)$, in which the local function of node $s$ is denoted
$f'_s$ and defined as

\begin{equation*}
	\forall x \in \B^{S \cup I},\ f'_s(x) = f_s( x \circ \hat \omega ),
	\text{ with } \hat\omega(i) = \begin{array}\{{ll}.
  		\omega(i) & \text{if } i \in \dom(\omega)\\
		i & \text{if } i \in I \setminus \dom(\omega)
	\end{array}\text{.}
\end{equation*}


\subsubsection{Output functions}

Output functions were first introduced in~\cite{C-Perrot2020} and present another
way of computing the evolution of an acyclic module. In the Boolean case, those
functions are defined on $\B^{I \times \{1, \ldots, D\}} \to \B$,
for I the input set of the module, and D some integer.
We interpret an input in $\B^{I \times \{1, \ldots, D\}}$ as an evaluation over
$\B$ of a set of variables ${I \times \{1, \ldots, D\}}$,
and for $\alpha \in I$ and $d \leq D$, we denote this variable by $\alpha_d$.
In the context of an acyclic module $M$, $\alpha_d$
is refering to the evaluation of the
input $\alpha$ on the $d$th update of the module. A vector
$j \in \B^{I \times \{1, \ldots, D\}}$ simply describes an evaluation of all
the inputs of the network over $D$ iterations. With such a vector, and
$x \in \B^S$, it is easy to see that the acyclic module $M$ can be updated
$k$ times in a row, for any $k \leq D$.
The result of this update is denoted by $M(x, j_{[1, \ldots, k]})$.
The \emph{delay} of an output function $O$ is the maximal value in the set of
all the $d \in \N$ for which there exists $\alpha \in I$ such that variable
$\alpha_d$ has an influence on the computation of $O$.
Finally,
for $M$ an acyclic module defined on the sets $S$ and $I$,
for $D$ a large enough integer, for $x \in \B^S$ and
$j \in \B^{I \times \{1, \ldots, D\}}$ some vectors, and for $s$ a node in $S$,
we define the output function of $s$, denoted $O_s$, as the output function
with minimal delay $d$ such that $O_s(j) = M(x, j_{[1, \ldots, k]})_s$.
Such a function always exists and is always unique.

\subsection{Promise problems}

In this paper, we make the hypothesis that every module passed as an input
of a complexity problem follows the property that each of its local functions
has only \emph{essential} variables. That is, a variable is included in the
circuit encoding that function if and only if the automaton or input represented
by that variable has an influence on said function.
This hypothesis will be implemented throughout this paper by the use of
promise problems~\cite{O-Goldreich2006}, which include a decision method which
can dismiss instances of the problem without that method's complexity cost
being included in the complexity of the problem.

This approach is motivated by the fact that obfuscating the relation
between automatons by building redundant variables in a circuit 
increases the complexity of most considered problems. We justify our
decision in two points: first, the approach of this paper is one of providing
and studying an applicable method in a context where misleading inputs in
local functions are unlikely. Second, despite the inclusion of these
promises
we find high complexity issues in our pipeline, and as such we consider that
it helps understanding the precise issues that
prevent our method from being efficient.

\section{From BANs to AMs} \label{s:PartA}


The first step of our process is to unfold a BAN into an AM. This simply
requires the removal of any cycle in the interaction digraph of the
BAN, and their replacement by inputs.
In the scope of this paper, the number of inputs generated is required to be
minimal. This is justified by the fact that the complexity of most of the
problems addressed in the pipeline highly depends on the number of inputs
of the considered AM.

\begin{fpproblem}
\fpproblemtitle{Acyclic Unfolding Functional Problem}
\fpprobleminput{A Boolean automata network $F$, an integer $k$.}
\fpproblempromise{The encoding of the local functions of $F$ only has
                  essential variables.}
\fpproblemoutput{An acyclic module $M$ with at most $k$ inputs and a
recursive wiring $\omega$ such that $\circlearrowright_\omega M = F$.}
\end{fpproblem}


\begin{theorem}
\label{th-unfold-up}
The Acyclic Unfolding Functional Problem is in \FNP.
\end{theorem}
\begin{proof}
  The promise of this problem allows us to compute the interaction digraph of
  $F$ in polynomial time.

  Consider the following simple non-deterministic algorithm: first guess a module
  $M$ and a wiring $\omega$; then check that the number of inputs in $M$ is no more than $k$
  and that $\circlearrowright_\omega M$ syntactically equals $F$.
  
  This algorithm operates in polynomial non-deterministic time since the
  recursive wiring is a simple substitution of variables,
  and thanks to the fact that one only needs to compare 
  $\circlearrowright_\omega M$ and $F$ at a syntactical level.
  Indeed, if any solution exists, then a solution exists with the
  same number of nodes, the same inputs, the same wirings, and such that
  the substitution operated by $\omega$ on $M$ leads to the local functions
  of $F$ written identically:
  because all local functions are equal on a semantic level,
  this is always possible,
  by starting from the local functions of $F$ and operating variable
  substitutions that are then reversed by the recursive wiring $\omega$
  (remark that this last ``reversed'' construction is not required to be
  computable in polynomial time).
%
\end{proof}

\begin{theorem}
\label{th-unfold-down}
The Acyclic Unfolding Functional Problem is \NP-hard.
\end{theorem}
\begin{proof}
Let us provide a reduction from the Feedback Vertex Set problem. We provide
$f$ a function that for any instance $(G, k)$ of the Feedback Vertex Set problem,
provides an instance $(F, k)$ of the Optimal Acyclic Unfolding Problem
where $S = V(G)$ and $f_s$ is a OR function of exactly every node $s'$ such
that $(s', s) \in A(G)$. This construction is explicitly designed so that
the interaction digraph of $F$ is isomorphic to $G$. Clearly $f$ is computable
in polynomial time.
We also provide $g$ a function that for $(G, k)$ an instance of the
Feedback Vertex Set problem, and $M$ a solution to the
Optimal Acyclic Unfolding Problem, checks if $S = V(G)$, and then
deduces the solution for $(G, k)$ the following way: $s$ is part of
the feedback vertex set if and only if its influence has been replaced by an
input
in $M$. This means that the
variable $x_s$ has been replaced in every local function in $M$ by the same
input variable.
Finally $g$ checks that the size of the obtained set is not greater
than $k$. It is clear that $g$ is polynomial.

From the definition of $f$ and $g$, it follows that
the Feedback Arc Set problem reduces in polynomial time to
the Optimal Acyclic Unfolding Problem, which implies the result.
\end{proof}

\begin{example}
\label{ex1-AM}
Consider $S_A$ and $F_A$ of Example~\ref{ex1-BAN-def}. Let us define
 $I_A = \{\alpha\}$. Let $M_A$ be the acyclic module that defines
$f'_a(x) = x_{\alpha}$,
$f'_b(x) = f'_c(x) = x_a$,
and $f'_d(x) = \neg x_b \vee \neg x_c$. The module $M_A$ is a valid answer to
the instance $F_A, k = 1$ of the Acyclic Unfolding Functional Problem.
The interaction digraph of this module is represented in Figure~\ref{f:ex:MOD}
(left panel).
\end{example}

\begin{example}
\label{ex2-AM}
Consider $S_B$ and $F_B$ of Example~\ref{ex2-BAN-def}. Let us define
$I_B = \{\alpha_{St}, \alpha_{Sl}, \alpha_{Sk},$
             $\alpha_{Pp}, \alpha_{C}, $ $\alpha_{C^*}\}$.
Let $M_B$ be the acyclic module that defines
$f'_{St}(x) = \neg x_{\alpha_{St}}$,
$f'_{Sl}(x) = \neg x_{\alpha_{Sl}} \vee x_{\alpha_{C^*}}$,
$f'_{Sk}(x) = x_{\alpha_{St}} \vee \neg x_{\alpha_{Sk}}$,
$f'_{Pp}(x) = x_{\alpha_{Sl}} \vee \neg x_{\alpha_{Pp}}$,
$f'_{Ru}(x) = f_{S9}(x) = \neg x_{\alpha_{Sk}} \vee x_{\alpha_{Pp}}
                    \vee \neg x_{\alpha_C} \vee \neg x_{\alpha_{C^*}}$,
$f'_C(x) = \neg x_{Ru} \vee \neg x_{S9} \vee \neg x_{\alpha_{Sl}}$,
$f'_{C25}(x) = \neg x_{\alpha_{Pp}} \vee x_{\alpha_C}$,
$f'_{M}(x) = x_{\alpha_{Pp}} \vee \neg x_{\alpha_C}$,
and $f'_{C^*}(x) = \neg x_{Ru} \vee \neg x_{S9} \vee x_{C25} \vee \neg x_M$.
The module $M_B$ is a valid answer to
the instance $F_B, k = 6$ of the Acyclic Unfolding Functional Problem.
The interaction digraph of this module is represented in
Figure~\ref{f:ex:MOD} (right panel).
\end{example}

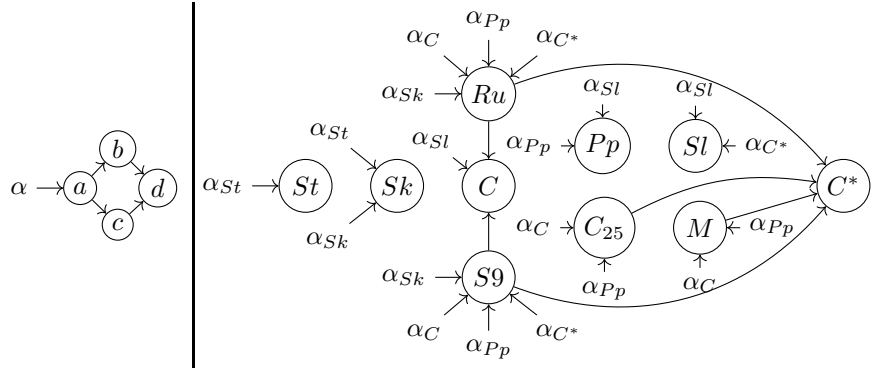
\begin{figure}[t!]
\begin{center}
\begin{minipage}{.2\textwidth}
\begin{tikzpicture} [->, node distance=0.15cm]

        \node (CENTER) at (0, 0) {};

        \node[state, inner sep=2pt,minimum size=0] (A) [left = of CENTER] {$a$};
        \node[state, inner sep=2pt,minimum size=0] (B) [above= of CENTER] {$b$};
        \node[state, inner sep=2pt,minimum size=0] (C) [below= of CENTER] {$c$};
        \node[state, inner sep=2pt,minimum size=0] (D) [right= of CENTER] {$d$};

        \path (A) edge (B) edge (C)
              (B) edge (D)
              (C) edge (D);

	    \node [left =10pt of A] (alpha1) {$\alpha$};
	    \draw (alpha1) -- (A);
		
\end{tikzpicture}
\end{minipage}
\vline
\begin{minipage}{.70\textwidth}
\begin{tikzpicture} [->, node distance=0.5cm]

        \node (CENTER) at (0, 0) {};

        \node[state, inner sep=2pt,minimum size=20pt] (C) [left = of CENTER] {$C$};
        \node[state, inner sep=2pt,minimum size=20pt] (Sk) [left= of C] {$Sk$};
        \node[state, inner sep=2pt,minimum size=20pt] (St) [left= of Sk] {$St$};
        \node[state, inner sep=2pt,minimum size=20pt] (Ru) [above= of C] {$Ru$};
        \node[state, inner sep=2pt,minimum size=20pt] (S9) [below= of C] {$S9$};
        \node[state, inner sep=2pt,minimum size=20pt] (Pp) [above right=0.2cm of CENTER] {$Pp$};
        \node[state, inner sep=2pt,minimum size=20pt] (C25) [below right=0.2cm of CENTER] {$C_{25}$};
        \node[state, inner sep=2pt,minimum size=20pt] (Sl) [right= of Pp] {$Sl$};
        \node[state, inner sep=2pt,minimum size=20pt] (M) [right= of C25] {$M$};
        \node[state, inner sep=2pt,minimum size=20pt] (C*) [right=4cm of C] {$C^*$};

        \path (Ru) edge[->, bend left = 35] (C*)
                   edge[->] (C)
              (S9) edge[->, bend right = 35] (C*)
                   edge[->] (C)
              (C25) edge [bend left = 18] (C*)
              (M) edge (C*);

         \node [left = 10pt of St] (Sta1) {$\alpha_{St}$};
         \draw (Sta1) -- (St);
         \node [above left = 10pt of Sk] (Ska1) {$\alpha_{St}$};
         \node [below left = 10pt of Sk] (Ska2) {$\alpha_{Sk}$};
         \draw (Ska1) -- (Sk);
         \draw (Ska2) -- (Sk);
         \node [above left = 10pt of Ru] (Rua1) {$\alpha_{C}$};
         \node [left = 10pt of Ru] (Rua2) {$\alpha_{Sk}$};
         \node [above right = 10pt of Ru] (Rua3) {$\alpha_{C^*}$};
         \node [above = 10pt of Ru] (Rua4) {$\alpha_{Pp}$};
         \draw (Rua1) -- (Ru);
         \draw (Rua2) -- (Ru);
         \draw (Rua3) -- (Ru);
         \draw (Rua4) -- (Ru);
         \node [below left = 10pt of S9] (S9a1) {$\alpha_{C}$};
         \node [left = 10pt of S9] (S9a2) {$\alpha_{Sk}$};
         \node [below right = 10pt of S9] (S9a3) {$\alpha_{C^*}$};
         \node [below = 10pt of S9] (S9a4) {$\alpha_{Pp}$};
         \draw (S9a1) -- (S9);
         \draw (S9a2) -- (S9);
         \draw (S9a3) -- (S9);
         \draw (S9a4) -- (S9);
         \node [above left = 5pt of C] (Ca1) {$\alpha_{Sl}$};
         \draw (Ca1) -- (C);
         \node [left = 5pt of Pp] (Ppa1) {$\alpha_{Pp}$};
         \node [above = 5pt of Pp] (Ppa2) {$\alpha_{Sl}$};
         \draw (Ppa1) -- (Pp);
         \draw (Ppa2) -- (Pp);
         \node [above = 5pt of Sl] (Sla1) {$\alpha_{Sl}$};
         \node [right = 5pt of Sl] (Sla2) {$\alpha_{C^*}$};
         \draw (Sla1) -- (Sl);
         \draw (Sla2) -- (Sl);
         \node [left = 5pt of C25] (C25a1) {$\alpha_{C}$};
         \node [below = 5pt of C25] (C25a2) {$\alpha_{Pp}$};
         \draw (C25a1) -- (C25);
         \draw (C25a2) -- (C25);
         \node [right = 5pt of M] (Ma1) {$\alpha_{Pp}$};
         \node [below = 5pt of M] (Ma2) {$\alpha_{C}$};
         \draw (Ma1) -- (M);
         \draw (Ma2) -- (M);
		
\end{tikzpicture}
\end{minipage}
\end{center}
\caption{
\label{f:ex:MOD}
On the left, the interaction digraph of $M_A$, as described in Example~\ref{ex1-AM}.
On the right, the interaction digraph of $M_B$, as described in Example~\ref{ex2-AM}.
}
\end{figure}

\section{Output functions} \label{s:PartB}

Output functions were first introduced in~\cite{C-Perrot2020}. They
are a way to characterise the asymptotic behaviour of an AM as a
set of Boolean functions that are computed from the local functions of the AM.
Computing the output functions of an AM is a crucial step in the pipeline
proposed in this work.

\begin{fpproblem}
\fpproblemtitle{Output Circuit Computation Problem}
\fpprobleminput{An acyclic module $M$, and $X \subseteq S$ a set of output nodes.}
\fpproblempromise{The encoding of the local functions of $M$ only has
                  essential variables.}
\fpproblemoutput
  {An output function for each node in $X$, encoded as a Boolean circuit.} 
\end{fpproblem}

\begin{theorem}
The Output Circuit Computation Problem is in FP.
\end{theorem}
\begin{proof}
The promise of this problem allows us to deduce
the interaction digraph of $M$ in polynomial time.

To compute $X$, we provide an algorithm to compute the output function circuit
of any node $s \in S$ in polynomial time.

The algorithm first constructs a list of requirements.
This list is initialy $R_0 = \{(s, 0)\}$, which can be interpreted to say that
we require the construction of the output
function of $s$ with added delay $0$.

We construct the next list the following way:
$(t', d') \in R_{k+1}$ if and only if there exists some $(t, d) \in R_{k}$ such
that $t'$ influences $t$ in $M$.

The total list $R$ is simply defined as $R = \bigcup_{i \in \N} R_i$.

\begin{claim}
$R$ is computable in polynomial time in the size of $M$.
\end{claim}

To see that this is true, consider that for $D$ the maximal depth of $M$,
the maximal $d$ such that $(t, d) \in R_k$ for any $k$ is $D$. Indeed since
the interaction digraph of $M$ is acyclic, the maximal delay value can only
be obtained by following the longest path in $M$. As such we can conclude that
the size of any $R_k$ is bounded by $D \times n$, for $n$ the size of the
network $M$. Finally, by a similar argument,
consider that the list $R$ converges after a maximum of $D$ steps. This implies
that the list $R$ is computed after $D$ steps of a $D \times n$ costly
process, and $R$ can therefore be computed in polynomial time.

We can construct the Boolean circuit from $R$ in the following way: for
every pair $(t, d) \in R$, take an instance of the Boolean circuit which encodes
the local function of $t$. Combine all of these instances the following way:
any input variable in $I$ is replaced by its delayed counterpart with delay
$1 + d$. For example, if a variable $\alpha$ appeared in the local function
of node $t$, substitute it by the variable $\alpha_{d + 1}$. Then, for any
gate displaying an input variable $t' \in S$, replace it with the same gate,
which rather than taking the value of variable $t'$, takes the value of the
output of the circuit that computes the local function of the node $t'$ with
added delay $d + 1$. By definition of $R$ this circuit will always be in $R$.
The obtained circuit computes the output function of the node $s$.

Repeat this process for every $s \in X$.
\end{proof}

\begin{example}
\label{ex1-OUT}
Consider $M_A$ of Example~\ref{ex1-AM}. Let
$X_A = \{d\}$ be an instance of the
Output Circuit Computation Problem. The circuit
$O_d = \neg \alpha_3$ is a
valid answer to that instance.
\end{example}

\begin{example}
\label{ex2-OUT}
Consider $M_B$ of Example~\ref{ex2-AM}. Let
$X_B = \{St, Sk, Sl, Pp, C, C^*\}$ be an instance of the
Output Circuit Computation Problem. The circuits
$O_{St} = \neg \alpha_{St, 1}$,
$O_{Sl} = \neg \alpha_{Sl, 1} \vee \alpha_{C^*, 1}$,
$O_{Sk} = \alpha_{St, 1} \vee \neg \alpha_{Sk, 1}$,
$O_{Pp} = \alpha_{Sl, 1} \vee \neg \alpha_{Pp, 1}$,
$O_C = ( \alpha_{Sk, 2} \wedge \neg \alpha_{Pp, 2}
  \wedge \alpha_{C, 2} \wedge \alpha_{C^*, 2} ) \vee \neg \alpha_{Sl, 1}$
and $O_{C^*} = \alpha_{C, 2} \vee \neg \alpha_{Pp, 2}$ taken altogether are a
valid answer to that instance.
\end{example}

\section{Optimal acyclic module synthesis} \label{s:PartC}


\subsection{Module Synthesis}

This part of the process takes in a set of output functions and generates a
module that realizes these functions with an hopefully minimal number of nodes.
In this part the actual optimisation of the pipeline, if any, can
be directly observed. It is also the part of the pipeline which is the most
computationnaly costly.

\begin{fproblem}
\fproblemtitle{Module Synthesis Problem}
\fprobleminput{A set $I$ of input labels, a finite set of output functions $O$,
encoded as Boolean circuits,
defined on those labels, and $k$ an integer.}
\fproblemoutput{An acyclic module $M$ with at most $k$ nodes such that every
function in $O$ is the output function of at least one node in $M$.} 
\end{fproblem}

\subsection{Complexity results}

\begin{theorem}
\label{th-syn-down}
The Module Synthesis Problem is coNP-hard.
\end{theorem}
\begin{proof}
Consider $f$ an instance of the Tautology problem, with $I$ the set of propositional
variables contained in $f$. We define $f'$ as the output function defined on the
labels $I$ such that $f'$ is obtained from $f$ by substituting all variables
$\alpha \in I$ by their equivalent of delay 1, $\alpha_1$. Let us also define
$f_1$ as the constant output function of delay $0$ which value is always $1$.
We compose an instance of the Module Synthesis Problem with $I$ the set of input
labels, $O = \{f', f_1\}$ and $k = 0$. This instance has a solution if and only
there exists an acyclic module with only one node such that the output function
of this node is equivalent to all the output functions in $O$. This implies that,
if the problem has a solution, $f'$ is equivalent to $f_1$, which proves that
$f'$ and $f$ are tautologies.
Therefore computing the output of the Module Synthesis
Problem requires solving a coNP-hard decision problem.
\end{proof}

\begin{theorem}
The Module Synthesis Problem is in \FNPcoNP.
\end{theorem}
\begin{proof}
Consider the following algorithm. First, guess an acyclic module $M$, with
size $k$. Compute every output function of the network, which is in FP. Then
simply check that every function in $O$ is equivalent to at least one
output function in $M$, which requires at most $|M| \times |O|$ calls to
a \coNP\ oracle.
\end{proof}

\subsection{Refining the complexity bounds}

It is unclear whether the synthesis problem can be proven to be in \FcoNP\ or
to be \NPcoNP-hard. An attempt has been made to prove the former by using a
greedy algorithm which would fuse nodes in an acyclic module, starting
from a trivially large enough module. This method requires solving the
following problem:

\begin{dproblem}
\problemtitle{Module Local Fusion Problem}
\probleminput{An acyclic module $M$ defined on sets $S$ and $I$,
  and $a, b$ two different nodes.}
\problemquestion{Is there a local function $f_c$ such that there exists
some acyclic module $M'$ defined on the node set
$S \cup \{c\} \setminus \{a, b\}$ and input set $I$, such that $f_c \equiv f'_c$
and $O_s \equiv O'_s$ for $s \in S \setminus \{a, b\}$?}
\end{dproblem}

This problem formalises the idea of replacing two nodes by one in an acyclic
module, such that every other output function in the module is conserved.
Assuming the removed nodes are not considered outputs of the network is
an important step of any greedy algorithm that would try to optimise the size
of an acyclic module.

It is rather simple to prove that this problem is coNP-hard, since its computation
requires checking the equivalence of multiple pairs of Boolean circuits. It
is also rather easy to see that it is in \NPcoNP, as one can guess $f_c$ and
$M'$ in polynomial time and verify the solution using a polynomial amount of
calls to a \coNP\ oracle, one for every equivalence check.

The function $f_c$ could be composed as a binary function of the results
of $f_a$ and $f_b$, as it is intuitive to suppose that assuming such a fusion
is possible, then every node influenced by $a$ or $b$ should be computable
from such a composition of $a$ and $b$. The issue however is that this process
requires modifying every node influenced by $a$ or $b$ such that their output
functions match the output functions in $M$.

It is unclear that there should exist a method in \coNP\ to ensure this
modification such that the output functions are conserved.
This leads us to
believe that a greedy algorithm wouldn't prove the Optimal Module Synthesis
Problem to be in \FcoNP.

Similarly, it is interesting to consider the open question of whether or not the
Module Synthesis Problem can be proven \NPcoNP-hard. This implies to prove,
between other things, that the problem is \NP-hard. This is, to us, another
open problem as the Module Synthesis Problem does not seem equiped to compute
the satisfaction of a Boolean formula or circuit.

\subsection{Similarities to other optimisation problems}

This open question bears strong ressemblance to another open problem that
concerns Boolean circuits. The Circuit Minimisation Problem is a problem that
asks to provide a Boolean circuit below a given size such that it computes
a Boolean function given as a truth table as the input of the
problem~\cite{C-Kabanets2000}. The problem is trivially in \NP\ but it is not
known whether the problem is in $\text{P}$ or \NP-hard, as both possibilities
imply
proving other results which seem beyond the currently known techniques.
The same problem has been found to be \NP-complete in both restricted (DNFs) and
generalised (unrestricted Boolean circuits)
variations of the Boolean circuit model~\cite{C-Ilango2020}.

There are strong similarities between acyclic modules and Boolean circuits.
Both are defined on acyclic digraphs, have inputs and outputs,
and compute Boolean functions. It is important to note that this analogy is
misleading when talking about the optimisation of their size. Optimising
a Boolean circuit requires the optimisation of a Boolean function in terms
of the number of gates that computes it. Optimising an acyclic module, however,
requires the optimisation of a network of functions with respect to a notion
of delay of the inputs,
whereas in this case one node may contain an arbitrary Boolean function.
As such these problems seem too independent
to provide any reduction between them.

\subsection{Examples}

\begin{example}
\label{ex1-SYN}
Consider the output function $O_d$ defined in Example~\ref{ex1-OUT}.
Let us define $M'_A$ as the module defined on $S'_A = \{a, b, d\}$ and
$I_A = \{\alpha\}$, such that
$f''_a = x_{\alpha}$,
$f''_b = x_a$ and $f_d = \neg x_b$.
The module $M'_A$ is a valid answer to the instance $I_A, \{O_d\}, k = 3$ of
the Module Synthesis Problem.
The interaction digraph of this module is depicted in Figure~\ref{f:ex:MODSYN}
(left panel).
\end{example}

\begin{example}
\label{ex2-SYN}
Consider the output functions $O_B = \{O_{St}, O_{Sl}, O_{Sk}, O_{Pp}, O_C, O_{C^*}\}$
defined in Example~\ref{ex2-OUT}.
Let us define $M'_B$ as the module defined on
$S'_B = \{St, Sl, Sk,$ $ Pp, Ru, C25, C^*\}$ and
$I_B = \{\alpha_{St}, \alpha_{Sl}, \alpha_{Sk}, \alpha_{Pp},
         \alpha_C, \alpha_{C^*}\}$, such that
$f''_{St}(x) = \neg x_{\alpha_{St}}$,
$f''_{Sl}(x) = \neg x_{\alpha_{Sl}} \vee x_{\alpha_{C^*}}$,
$f''_{Sk}(x) = x_{\alpha_{St}} \vee \neg x_{\alpha_{Sk}}$,
$f''_{Pp}(x) = x_{\alpha_{Sl}} \vee \neg x_{\alpha_{Pp}}$,
$f''_{Ru}(x) = \neg x_{\alpha_{Sk}} \vee x_{\alpha_{Pp}}
                    \vee \neg x_{\alpha_C} \vee \neg x_{\alpha_{C^*}}$,
$f''_C(x) = \neg x_{Ru} \vee \neg x_{\alpha_{Sl}}$,
$f''_{C25}(x) = \neg x_{\alpha_{Pp}} \vee x_{\alpha_C}$,
and $f'_{C^*}(x) = x_{C25}$.
The module $M'_B$ is a valid answer to the instance $I_B, O_B, k = 8$ of
the Module Synthesis Problem.
The interaction digraph of this module is depicted in Figure~\ref{f:ex:MODSYN}
(right panel).
\end{example}

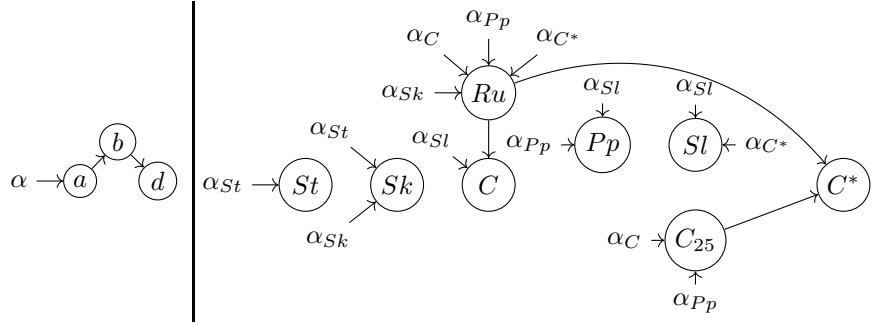
\begin{figure}[t!]
\begin{center}
\begin{minipage}{.2\textwidth}
\begin{tikzpicture} [->, node distance=0.15cm]

        \node (CENTER) at (0, 0) {};

        \node[state, inner sep=2pt,minimum size=0] (A) [left = of CENTER] {$a$};
        \node[state, inner sep=2pt,minimum size=0] (B) [above= of CENTER] {$b$};
        \node[state, inner sep=2pt,minimum size=0] (D) [right= of CENTER] {$d$};

        \path (A) edge (B)
              (B) edge (D);

	    \node [left =10pt of A] (alpha1) {$\alpha$};
	    \draw (alpha1) -- (A);
		
\end{tikzpicture}
\end{minipage}
\vline
\begin{minipage}{.7\textwidth}
\begin{tikzpicture} [->, node distance=0.5cm]

        \node (CENTER) at (0, 0) {};

        \node[state, inner sep=2pt,minimum size=20pt] (C) [left = of CENTER] {$C$};
        \node[state, inner sep=2pt,minimum size=20pt] (Sk) [left= of C] {$Sk$};
        \node[state, inner sep=2pt,minimum size=20pt] (St) [left= of Sk] {$St$};
        \node[state, inner sep=2pt,minimum size=20pt] (Ru) [above= of C] {$Ru$};
        \node[state, inner sep=2pt,minimum size=20pt] (Pp) [above right=0.2cm of CENTER] {$Pp$};
        \node[state, inner sep=2pt,minimum size=20pt] (C25) [below= of Sl] {$C_{25}$};
        \node[state, inner sep=2pt,minimum size=20pt] (Sl) [right= of Pp] {$Sl$};
        \node[state, inner sep=2pt,minimum size=20pt] (C*) [right=4cm of C] {$C^*$};

        \path (Ru) edge[bend left = 35] (C*)
                   edge (C)
              (C25) edge (C*);

         \node [left = 10pt of St] (Sta1) {$\alpha_{St}$};
         \draw (Sta1) -- (St);
         \node [above left = 10pt of Sk] (Ska1) {$\alpha_{St}$};
         \node [below left = 10pt of Sk] (Ska2) {$\alpha_{Sk}$};
         \draw (Ska1) -- (Sk);
         \draw (Ska2) -- (Sk);
         \node [above left = 10pt of Ru] (Rua1) {$\alpha_{C}$};
         \node [left = 10pt of Ru] (Rua2) {$\alpha_{Sk}$};
         \node [above right = 10pt of Ru] (Rua3) {$\alpha_{C^*}$};
         \node [above = 10pt of Ru] (Rua4) {$\alpha_{Pp}$};
         \draw (Rua1) -- (Ru);
         \draw (Rua2) -- (Ru);
         \draw (Rua3) -- (Ru);
         \draw (Rua4) -- (Ru);
         \node [above left = 5pt of C] (Ca1) {$\alpha_{Sl}$};
         \draw (Ca1) -- (C);
         \node [left = 5pt of Pp] (Ppa1) {$\alpha_{Pp}$};
         \node [above = 5pt of Pp] (Ppa2) {$\alpha_{Sl}$};
         \draw (Ppa1) -- (Pp);
         \draw (Ppa2) -- (Pp);
         \node [above = 5pt of Sl] (Sla1) {$\alpha_{Sl}$};
         \node [right = 5pt of Sl] (Sla2) {$\alpha_{C^*}$};
         \draw (Sla1) -- (Sl);
         \draw (Sla2) -- (Sl);
         \node [left = 5pt of C25] (C25a1) {$\alpha_{C}$};
         \node [below = 5pt of C25] (C25a2) {$\alpha_{Pp}$};
         \draw (C25a1) -- (C25);
         \draw (C25a2) -- (C25);
		
\end{tikzpicture}
\end{minipage}
\end{center}
\caption{
\label{f:ex:MODSYN}
On the left, the interaction digraph of $M'_A$, as described in Example~\ref{ex1-SYN}.
On the right, the interaction digraph of $M'_B$, as described in Example~\ref{ex2-SYN}.
}
\end{figure}

\section{Final wiring and analysis} \label{s:PartD}

The final step in the pipeline is simply to wire the module obtained in
Section~\ref{s:PartC} so that the obtained networks hold isomorphic attractors
to the input network. This is ensured by application of the following
result..

\begin{theorem}[\cite{C-Perrot2020}]
\label{th-limit-from-output}
	Let $M$ and $M'$ be two acyclic modules, with $T$ and $T'$ subsets of their nodes
	such that $|T| = |T'|$. If there exists $g$ a bijection from $I$ to $I'$ and
	$h$ a bijection from $T$ to $T'$ such that for
	every $s \in T$, $O_s$ and $O'_{h(s)}$ have same delay, and for every input
	sequence $j$ with length the delay of $O_s$,
	\[O_s(j) = O'_{h(s)}(j \circ g^{-1})\]
	then for any function $\omega : I \to T$, the networks
	$\circlearrowright_{\omega} M$ and $\circlearrowright_{h \circ \omega \circ g^{-1}} M'$ have
        isomorphic attractors (up to the renaming of automata given by $h$).
\end{theorem}

Applying this theorem to the current problem is simple: the module $M$ is the
module obtained in Section~\ref{s:PartA}, and the module $M'$ is the
module obtained in Section~\ref{s:PartC}. The set $T$ is the set of nodes
which are substituted by new inputs in the process described in
Section~\ref{s:PartA}.
The set $T'$ is the set of nodes in $M'$ which are considered as the output
of the module, for example when the module $M'$ is obtained as the result
of the application of the functional problem defined in Section~\ref{s:PartC}.

As modules $M$ and $M'$ are defined over the same set of inputs, the
bijection $g$ is the identity. The bijection $h$ is directly constructed so
that for all $s \in T$, $h(s)$ in $M'$ has an equivalent output function as
$s$ in $M$, which is always possible thanks to the careful structure of our
pipeline. It follows quite clearly that for any $s \in T$, and for
any input sequence $j$, $O_s(j) = O'_{h(s)}(j \circ g^{-1})$ holds, and
the theorem applies.

\begin{example}
\label{ex1-FIN}
Consider $M'_A$ of Example~\ref{ex1-SYN}. Let $\omega_A(\alpha) = d$. The AN
$\circlearrowright_{\omega_A} M'_A$
is defined over $S'_A = \{a, b, d\}$ such that
$f'''_a(x) = x_d$,
$f'''_b(x) = x_a$,
$f'''_d(x) = \neg x_b$.
The interaction digraph of this module is depicted in Figure~\ref{f:ex:FIN}
(left panel).
\end{example}

\begin{example}
\label{ex2-FIN}
Consider $M'_B$ of Example~\ref{ex2-SYN}. Let $\omega_B(\alpha_s) = s$,
for all $s \in X_B$.
The AN $\circlearrowright_{\omega_B} M'_B$
is defined over $S'_B = \{St, Sl, Sk, Pp, Ru, C25, C^*\}$ such that
$f'''_{St}(x) = \neg x_{St}$,
$f'''_{Sl}(x) = \neg x_{Sl} \vee x_{C^*}$,
$f'''_{Sk}(x) = x_{St} \vee \neg x_{Sk}$,
$f'''_{Pp}(x) = x_{Sl} \vee \neg x_{Pp}$,
$f'''_{Ru}(x) = \neg x_{Sk} \vee x_{Pp}
                    \vee \neg x_C \vee \neg x_{C^*}$,
$f'''_C(x) = \neg x_{Ru} \vee \neg x_{Sl}$,
$f'''_{C25}(x) = \neg x_{Pp} \vee x_C$,
and $f'_{C^*}(x) = x_{C25}$.
The interaction digraph of this module is depicted in Figure~\ref{f:ex:FIN}
(right panel).
\end{example}

This allows us to compute the attractors of any BAN by computing the dynamics
of another BAN with possibly less nodes, thus dividing the number of computed
configurations by some power of two. Examples throughout this paper showcase
the application of the pipeline over two initial examples.

Examples~\ref{ex1-BAN-def}, \ref{ex1-AM}, \ref{ex1-OUT}, \ref{ex1-SYN} and
\ref{ex1-FIN} show the optimisation of a simple four nodes network into a three
nodes equivalent network.
The optimisation proceeds here by \lq compacting\rq\ two
trivially equivalent nodes, $b$ and $c$, into one. The resulting BAN has
dynamics $2^1$ times smaller than the initial network, with isomorphic attractors.
Examples~\ref{ex2-BAN-def}, \ref{ex2-AM}, \ref{ex2-OUT}, \ref{ex2-SYN} and
\ref{ex1-FIN} show the optimisation of a larger, more intricate network which
is drawn from a model predicting the cell cycle sequence of fission
yeast~\cite{J-Davidich2008}. This practical example, processed through our
pipeline, reduces from $10$ nodes to $8$. This implies a reduction in dynamics
size of $2^2$, while keeping isomorphic attractors.
Both sets of examples are illustrated throughout the paper in
Figures~\ref{f:ex:BAN}, \ref{f:ex:MOD}, \ref{f:ex:MODSYN} and \ref{f:ex:FIN}.

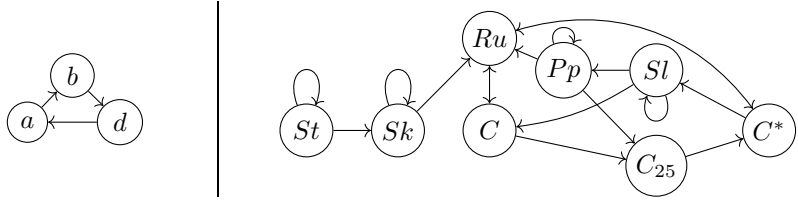
\begin{figure}[t!]
\begin{minipage}{.3\textwidth}
\begin{center}
\begin{tikzpicture} [->, node distance=0.2cm]

        \node (CENTER) at (0, 0) {};

        \node[state, inner sep=3pt,minimum size=0] (A) [left = of CENTER] {$a$};
        \node[state, inner sep=3pt,minimum size=0] (B) [above= of CENTER] {$b$};
        \node[state, inner sep=3pt,minimum size=0] (D) [right= of CENTER] {$d$};

        \path (A) edge (B)
              (B) edge (D)
              (D) edge (A);
		
\end{tikzpicture}
\end{center}
\end{minipage}
\vline
\begin{minipage}{.7\textwidth}
\begin{center}
\begin{tikzpicture} [->, node distance=0.5cm]

        \node (CENTER) at (0, 0) {};

        \node[state, inner sep=2pt,minimum size=20pt] (C) [left = of CENTER] {$C$};
        \node[state, inner sep=2pt,minimum size=20pt] (Sk) [left= of C] {$Sk$};
        \node[state, inner sep=2pt,minimum size=20pt] (St) [left= of Sk] {$St$};
        \node[state, inner sep=2pt,minimum size=20pt] (Ru) [above= of C] {$Ru$};
        \node[state, inner sep=2pt,minimum size=20pt] (Pp) [above=0.3cm of CENTER] {$Pp$};
        \node[state, inner sep=2pt,minimum size=20pt] (Sl) [right= of Pp] {$Sl$};
        \node[state, inner sep=2pt,minimum size=20pt] (C25) [below= of Sl] {$C_{25}$};
        \node[state, inner sep=2pt,minimum size=20pt] (C*) [right=3cm of C] {$C^*$};

        \path (St) edge (Sk)
              (Sk) edge (Ru)
              (Ru) edge[<->] (C)
                   edge[<->, bend left = 35] (C*)
              (C) edge (C25)
              (Pp) edge (Ru) edge (C25)
              (Sl) edge (Pp)
                   edge [bend left = 12] (C)
              (C25) edge (C*)
              (C*) edge (Sl);

        \draw (St) to [out=110, in=70, looseness=7] (St);
	    \draw (Sk) to [out=110, in=70, looseness=7] (Sk);
	    \draw (Pp) to [out=110, in=70, looseness=3] (Pp);
	    \draw (Sl) to [out=290, in=250, looseness=5] (Sl);
		
\end{tikzpicture}
\end{center}
\end{minipage}
\caption{
\label{f:ex:FIN}
On the left, the interaction digraph of $F'_A$, as described in Example~\ref{ex1-FIN}.
On the right, the interaction digraph of $F'_B$, as described in Example~\ref{ex2-FIN}.
}
\end{figure}



\section{Conclusion}

The present paper showcases an innovative way of reducing the cost of computing
the attractors of Boolean automata networks.
The method provides better optimisation
on networks showing structural redundancies, which are removed by the pipeline.
The limitations of this method
are still significant; it requires solving a problem that is at
least \coNP-hard, and believed to be \FNPcoNP-complete.
As it presently stands, this method is not as much a convincing practical tool as it is
a good argument in favor of the powerfulness of acyclic modules, their
output functions, and the approaches they allow together towards the
computation of BAN dynamics.

Future perspectives include finding better complexity bounds to the
Module Synthesis Problem, and generalising the formalism of output functions 
and the optimisation pipeline to
different update schedules distinct from parallel.


\bibliographystyle{plain}
{\small{\bibliography{pipeline}}}


\end{document}